\documentclass[11pt,A4paper,fullpage]{article}

\usepackage{amsfonts}
\usepackage{fullpage}
\usepackage{graphicx,color}
\usepackage{xspace}
\usepackage{enumerate}
\usepackage{amssymb,amsmath}
\usepackage{amsthm}
\usepackage{subfig}

\usepackage[T1]{fontenc}
\usepackage[utf8]{inputenc}
\usepackage{authblk}
\usepackage{hyperref}

\newtheorem{theorem}{Theorem}

\newtheorem{corollary}[theorem]{Corollary}
\newtheorem{proposition}[theorem]{Proposition}
\newtheorem{lemma}[theorem]{Lemma}

\newcommand{\comp}{\ensuremath{\textrm{cr}}\xspace}

\begin{document}

\author[1]{Spyros Angelopoulos\footnote{This research benefited from support of the FMJH Program PGMO and from the support to this program from Thales-EDF-Orange.}}

\date{}

\affil[1]{CNRS and Sorbonne   Universit\'e,  Laboratoire   d’Informatique  de  Paris   6, 4 Place Jussieu, Paris, France 75252.
{\tt spyros.angelopoulos@lip6.fr}
}

\title{Online Search With a Hint}

\maketitle

\begin{abstract}
The {\em linear search} problem, informally known as the {\em cow path} problem, is one of the fundamental problems in search theory. In this problem, an immobile target is hidden at some unknown position on an unbounded line, and a mobile searcher, initially positioned at some specific point of the line called the {\em root}, must traverse the line so as to locate the target. The objective is to minimize the worst-case ratio of the distance traversed by the searcher to the distance of the target from the root, which is known as the {\em competitive ratio} of the search. 

In this work we study this problem in a setting in which the searcher has a {\em hint} concerning the target.
 We consider three settings in regards to the nature of the hint: i) the hint suggests the exact position of the target on the line; ii) the hint suggests the direction of the optimal search (i.e., to the left or the right of the root); and iii) the hint is a general $k$-bit string that encodes some information concerning the target. Our objective is to study the {\em Pareto}-efficiency of strategies in this model. Namely, we seek optimal, or near-optimal tradeoffs between the searcher's performance if the hint is correct (i.e., provided by a trusted source) and if the hint is incorrect (i.e., provided by an adversary). 
\end{abstract}


\section{Introduction}
\label{sec:intro}

Searching for a target is a common task in everyday life, and an important computational problem with numerous applications.  Problems involving search arise in such diverse areas as drilling for oil in multiple sites, the forest service looking for missing backpackers, search-and-rescue operations in the open seas, and navigating a robot between two points on a terrain.  All these problems involve a mobile {\em searcher} 
which must locate an immobile {\em target}, often also called {\em hider}, that lies in some unknown point in the {\em search domain}, i.e, the environment in which the search takes place. The searcher starts from some initial placement within the domain, denoted by $O$, which we call the {\em root}. There is, also, some underlying concept of quality of search, in the sense that we wish, in informal terms, for the searcher to be able to locate the target as efficiently as possible.

One of the simplest, yet fundamental search problems is searching on an infinite line that is unbounded both to the left and to the right of the root. In this problem, which goes back to Bellman~\cite{bellman} and
Beck and Newman~\cite{beck:ls}, the objective is to find a search strategy that minimizes the {\em competitive ratio} of search. More precisely, 
let $S$ denote the {\em search strategy}, i.e., the sequence of moves that the searcher performs on the line. Given a target $t$, let $d(S,t)$ denote the total distance that the searcher has traveled up to the time it locates the target, and $d(t)$ the distance of $t$ from $O$.
We define the competitive ratio of $S$ as
\[
\comp(S) =\sup_t \frac{d(S,t)}{d(t)}.
\]
A strategy of minimum competitive ratio is called {\em optimal}. The problem of optimizing the competitive ratio of search on the line is known as the {\em linear search} problem (mostly within Mathematics and Operations Research), but is also known in Computer Science as the {\em cow path} problem. 

It has long been known that the optimal (deterministic) competitive ratio of linear search is 9~\cite{beck:yet.more}, and is derived by a simple {\em doubling} strategy. Specifically, let the two semi-infinite branches of the line be labeled with $0,1$ respectively. Then in iteration $i$, with $i \in \mathbb N$, the searcher starts from $O$, traverses branch $i \bmod 2$ to distance $2^i$, and returns to the root. 

Linear search, and its generalization, the $m$-ray search problem, in which the search domain consists of $m$ semi-infinite branches have been studied in several settings. Substantial work on linear search was done in the '70s and '80s predominantly by Beck and Beck, see e.g.,~\cite{beck1965more,beck1973return,beck1984linear,beck1986linear,beck1992revenge}. Gal showed that a variant of the doubling strategy is optimal for $m$-ray search~\cite{gal:general,gal:minimax}. These results were later rediscovered and extended in~\cite{yates:plane}.

Other related work includes the study of randomization~\cite{schuierer:randomized} and~\cite{ray:2randomized}; multi-searcher strategies~\cite{alex:robots}; searching with turn cost~\cite{demaine:turn,Angelopoulos2017};
the variant in which some probabilistic information on the target is
known~\cite{jaillet:online, informed.cows}; the related problem of designing {\em hybrid algorithms}~\cite{hybrid}; searching with an upper bound on the
distance of the target from the root~\cite{ultimate} and~\cite{revisiting:esa}; fault tolerant search~\cite{DBLP:conf/isaac/CzyzowiczGKKNOS16,kupavskii2018lower}; and performance measures beyond the competitive ratio~\cite{hyperbolic,oil,DBLP:conf/stacs/0001DJ19}. 
Competitive analysis has been applied beyond the linear and star search, for example in searching within a 
graph~\cite{koutsoupias:fixed,fleischer:online,stacs-expanding,ANGELOPOULOS2020781}.

\subsection{Searching with a hint}

Previous work on competitive analysis of deterministic search strategies has mostly assumed that the searcher has no information about the target, whose position is adversarial to the search. In practice, however, we expect that the searcher may indeed have some information concerning the target. For instance, in a search-and-rescue mission, there may be some information on the last sighting of the missing person, or the direction the person had taken when last seen. The question then is: how can the searcher leverage such information, and to what possible extent?

If the hint comes from a source that is trustworthy, that is, if the hint is guaranteed to be correct, then the performance of search can improve dramatically. For example, in our problem, if the hint is the branch on which the target lies, then the optimal search is to explore that branch until the target is found, and the competitive ratio is 1. There is, however, an obvious downside: if the hint is incorrect, the search may be woefully inefficient since the searcher will walk eternally on the wrong branch, and the competitive ratio in this case is unbounded.

We are thus interested in analyzing the efficiency of search strategies in a setting in which the hint may be compromised. To this end, we first need to define formally the concept of the hint, as well as an appropriate performance measure for the search strategy. In general, the hint $h$ is a binary string of size $k$, where the $i$-th bit is a response to a {\em query} $Q_i$. For example, one can define a single query $Q$ as ``Is the target within distance at most 100 from $O$?'' and  a one-bit hint, so that the hint answers a range query. For another example, if $Q$=``Is the target to the left or to the right of $O$?'', then a 1-bit hint informs the searcher about the direction it should pursue. From the point of view of upper bounds (positive results), we are interested in settings in which the queries and the associated hints have some natural interpretation, such as the ones given above. From the point of view of lower bounds (impossibility results), we are interested on the limitations of general $k$-bit hint strings which may be associated with {\em any} query, as we will discuss in more detail later.

Concerning the second issue, namely evaluating the performance of a search strategy $S$ with a hint $h$,
note first that $S$ is a function of $h$. We will analyze the competitiveness of $S(h)$ in a model in which the competitive ratio is not defined by a single value, but rather by a pair $(c_{S,h},r_{S,h})$. The value $c_{S,h}$ describes the competitive ratio of $S(h)$ assuming that $h$ is {\em trusted}, and thus guaranteed to be correct. The value $r_{S,h}$ describes the competitive ratio of $S(h)$ when the hint is given by an {\em adversarial} source. 
More formally, we define
\begin{equation}
c_{S,h}= \sup_{t} \inf_{h} \frac{d(S(h),t,h)}{d(t)} ,
      \quad  \textrm{and} \quad
r_{S,h} = \sup_{t} \sup_{h} \frac{d(S(h),t,h)}{d(t)},
\label{eq:definition}
\end{equation} 
where $d(S(h),t,h)$ denotes the distance traversed in $S(h)$ for locating a target $t$ with a hint $h$. We will call $c_{S,h}$ the {\em consistency} of $S(h)$, and $r_{S,h}$ the {\em robustness} of $S(h)$. To simplify notation, we will often write $S$ instead of $S(h)$ when it is clear from context that we refer to a strategy with a hint $h$.

For example, if the hint $h$ is the branch on which the target lies, then the strategy that always trusts the hint is $(1,\infty)$ competitive, whereas the strategy that ignores the hint entirely is $(9,9)$-competitive. Our objective is then to find strategies that are provably {\em Pareto-optimal} or {\em Pareto-efficient} in this model, and thus identify the strategies with the best tradeoff between robustness and consistency.

Our model is an adaptation, to search problems, of the untrusted advice framework for online algorithms 
proposed by Angelopoulos {\em et al.}~\cite{DBLP:conf/innovations/0001DJKR20}. In their work, the online algorithm is given some additional information, or {\em advice} which may, or may not be correct.
To the best of our knowledge, our setting is a first attempt to quantify, in an adversarial setting, the impact of general types of {\em predictions} in search games. It is also in line  with recent advances on improving the performance of online algorithms using predictions, such as the work of Lykouris and Vassilvitskii~\cite{DBLP:conf/icml/LykourisV18}, who introduced the concepts of consistency and robustness in the context of paging, and the work of Purohit {\em et al.}~\cite{NIPS2018_8174}, who applied it to general online problems. Our framework for the $k$-bit hint is also related to other work in Machine Learning, such as clustering with $k$ noisy queries, e.g., the work of Mazumdar and Saha~\cite{NIPS2017_7161}.

It should be emphasized that there is previous work that has studied the impact of specific types of hints on the performance of search strategies, such as bounds on the maximum or the minimum distance of the target from the root, e.g.,~\cite{ultimate,revisiting:esa,jaillet:online}. However, the hint in these works is always assumed to be trusted and correct.

\subsection{Contribution}

In this work we study the power of limitations of linear search with hints. 
Let $r \geq 9$ be a parameter that in general will denote the robustness of a search strategy, 
and let $b_r$ be defined as 
\[
b_r= \frac{\rho_r+\sqrt{\rho_r^2-4\rho_r}}{2} ,\quad \textrm{where $\rho_r=(r-1)/2$}.
\]
We consider the following classes of hints:

\medskip
\noindent
{\em $\bullet$ \ The hint is the position of the target.} Here, the hint describes the exact location of the target on the line: its distance from $O$, along with the branch (0 or 1) on which it lies. We present a strategy that is  $(\frac{b_r-1}{b_r+1}, r)$-competitive, and we prove it is Pareto-optimal.

\medskip
\noindent
{\em $\bullet$ \ The hint is the branch on which the target lies.} Here, the hint is information on whether the searcher is to the left or to the right of the root. We present a strategy that, given parameters $b>1$ and 
$\delta\in(0,1)$, has consistency
$c=1+2\cdot (\frac{b^2}{b^2-1}+ \delta \frac{b^3}{b^2-1})$, and robustness 
$r=1+2\cdot ( \frac{b^2}{b^2-1}+\frac{1}{\delta} \frac{b^3}{b^2-1})$.
Again, we prove that this strategy is Pareto-optimal. 

\medskip
\noindent
{\em $\bullet$ \ The hint is a general $k$-bit string.} In the previous settings, the hint is a single bit, which answers the corresponding query. Here we address the question: how powerful can be a single-bit hint, or more generally a $k$-bit hint? In other words, how powerful can $k$ binary queries be for linear search?
We give several upper and lower bounds on the competitiveness of strategies in this setting. First, we look at the case of a single-bit hint. Here, we give a 9-robust strategy that has consistency at most $1+4\sqrt{2}$, whereas we show that no 9-robust strategy can have conistency less than 5, for {\em any} associated query. For general robustness $r$, we give upper and lower bounds that apply to some specific, but broadly used class of strategies, including {\em geometric} strategies (see Section~\ref{sec:prelim} for a definition and Theorem~\ref{thm:asymptotic} for the statement of the result). For general $k$, and for a given $r\geq 9$, we give an $r$-robust strategy whose consistency decreases rapidly as function of $k$ (Proposition~\ref{prop:k.upper}).
\medskip

In terms of techniques, for the first setting described above (in which the hint is the position of the target), the main idea is to analyze a geometric strategy with ``large'' base, namely $b_r$, for $r\geq 9$. The technical difficulty here is the lower bound; to this end, we prove a lemma that shows, intuitively, that for any $r$-robust strategy, the search length of the $i$-ith iteration cannot be too big compared to the previous search lengths (Lemma~\ref{lemma:helper.bb}). This technical result may be helpful in more broad settings (e.g., we also apply it in the setting in which the advice is a general $k$-bit string).

Concerning the second setting, in which the hint describes the branch, we rely on tools developed by Schuierer~\cite{schuierer:lower} for lower-bounding the performance of search strategies; more precisely on a theorem for lower-bounding the supremum of a sequence of functionals. But unlike~\cite{schuierer:lower}, we use the theorem in a {\em parameterized} manner, that allows us to express the tradeoffs between the consistency and the robustness of a strategy, instead of their average. 

Concerning the third, and most general setting, our upper bounds (i.e., the positive results) come from a strategy that has a natural interpretation: it determines a partition of the infinite line into $2^k$ {\em subsets}, and the hint describes the partition in which the target lies. The lower bounds (negative results) come from information-theoretic arguments, as is typical in the field of {\em advice complexity} of online 
algorithms (see, e.g., the survey~\cite{Boyar:survey:2016}). 

The broader objective of this work is to initiate the study of search games with some limited, but potentially untrusted information concerning the target. As we will show, the problem becomes challenging even in a simple search domain such as the infinite line. The framework should be readily applicable to other search games, and the analysis need not be confined to the competitive ratio, or to worst-case analysis. For example, search games in bounded domains are often studied assuming a probability distribution on the target, with the objective to minimize the expected search time (for several such examples see the book~\cite{searchgames}).
However, very little work has addressed the setting in which the searcher may have access to hints, such as the
{\em High-Low} search games described in Section 5.2 of~\cite{searchgames}, in which a searcher wants to locate a hider on the unit interval by a sequence of guesses. Again, our model is applicable, in that one would like to find the best tradeoff on the expected time to locate the target assuming a trusted or untrusted hint.

\section{Preliminaries}
\label{sec:prelim}

In the context of searching on the line, a search strategy $X$ can be defined as an infinite set of pairs
$(x_i,s_i)$, with $i \in \mathbb N$, $x_i \in \mathbb{R}_{\geq 1}$ and $s_i\in\{0,1\}$. We call $i$
an {\em iteration} and $x_i$ the {\em length} of the $i$-th search {\em segment}. 
More precisely, in the $i$-th iteration, the searcher starts from the root $O$, traverses branch $s_i \bmod 2$ up to distance $x_i$ from $O$, then returns to $O$. It suffices to focus on strategies for which $x_{i+2} \geq x_i$, i.e., in any iteration the searcher always searches a new part of the line. We will sometimes omit the $s_i$'s from the definition of the strategy, if the direction is not important, i.e., the searcher can start by moving either to the left or to the right of $O$. In this case, there is the implicit assumption that $s_i$ and $s_{i+1}$ have complementary parities, since any strategy that revisits the same branch in consecutive iterations can be transformed to another strategy that is no worse, and upholds the assumption. We make the standing assumption that the target lies within distance at least a fixed value, otherwise every strategy has unbounded competitive ratio. In particular, we will assume that $t$ is such that $d(t) \geq 1$.

Given a strategy $X=(x_0,x_1, \ldots)$ (which we will denote by $X=(x_i)$, for brevity), its competitive ratio
is given by the expression 
\begin{equation}
\comp(X)=1+2\sup_{i \geq 0} \frac{\sum_{j=1}^{i}x_{j}}{x_{i-1}},
\label{eq:comp.X}
\end{equation}
where $x_{-1}$ is defined to be equal to 1. This expression is obtained by considering all the worst-case positions of the target, namely immediately after the turn point of the $i$-th segment (see e.g.,~\cite{schuierer:lower}).

Geometric sequences are important in search problems, since they often lead to efficient, or optimal strategies
(see, e.g., Chapters 7 and 9 in~\cite{searchgames}). We call the search strategy $G_b=(b^i)$ {\em geometric with base $b$}. From~\eqref{eq:comp.X}, we obtain that
\begin{equation}
\comp(G_b)=1+2\frac{b^2}{b-1}.
\label{eq:comp.geometric}
\end{equation}
For example, for the standard doubling strategy in which $x_i=2^i$, hence $b=2$, the above expression implies a competitive ratio of 9.

For any $r\geq 4$ define $\rho_r$ to be such that $r=1+2\rho_r$, thus $\rho_r=(r-1)/2$. Moreover, from~\eqref{eq:comp.geometric} and the definition of $b_r$, we have that $\comp(G_{{b_r}})=r$.

In the context of searching with a hint, we will say that a strategy is $(c,r)$-competitive if it has consistency at most $c$ and robustness at most $r$; equivalently we say that the strategy is $c$-consistent and $r$-robust. 
Clearly, an $r$-robust strategy gives rise to a strategy with no hints, and with competitive ratio at most $r$.

We conclude with some definitions that will be useful in Section~\ref{sec:direction}.
Let $X=(x_0,x_1,\ldots)$ denote a sequence of positive numbers. We define $\alpha_X$ as
\[
\alpha_X =\overline{\lim}_{n \rightarrow \infty} x_n^{1/n}.
\]
We also define as $X^{+i}$ the subsequence of $X$ starting at $i$, i.e, $X^{+i}=(x_i,x_{i+1},\ldots)$.
Last, we define the sequence 
$G_b(\gamma_0, \ldots \gamma_{n-1})$ as
\[
G_b(\gamma_0, \ldots \gamma_{n-1})= (\gamma_0, \gamma_1 a, \gamma_2 a^2, \ldots \gamma_{n-1} a^{n-1}, \gamma_0 a^n, \gamma_1 a^{n+1}, \ldots ).
\]

\section{Hint is the position of the target}
\label{sec:position}

In this section we study the setting in which the hint is related to the exact position of the target. Namely, the hint $h$ describes the distance $d(t)$ of the target $t$ from the root, as well as the branch on which it hides.  For any $r \geq 9$, we will give a strategy that is $(\frac{b_r+1}{b_r-1}, r))$-competitive. Moreover, we will show that this is Pareto-optimal. We begin with the upper bound.

\begin{theorem}
For any $r\geq 9$ there exists a $(\frac{b_r+1}{b_r-1}, r)$-competitive strategy for linear search in which the hint is the position of the target.
\label{thm:position.upper}
\end{theorem}

\begin{proof}
From the hint $h$, we have as information the distance $d(t)$ as well as the branch $t$ on which the target $t$ lies; without loss of generality, suppose that this branch is the branch 0. Recall that this information may or
may not be correct, and the searcher is oblivious to this. 

Consider the geometric strategy $G_{b_r}=(b_r^i)$, with $i\in \mathbb{N}$, and recall that 
$G_{b_r}$ is $r$-robust (as discussed in Section~\ref{sec:prelim}). There must exist an index $j_t$ such that 
\[
b_r^{j_t-2} < d(t) \leq b_r^{j_t},
\]
Define $\lambda=b_r^{j_t}/d(t) \geq 1$,  and let $G'$ denote the strategy 
\[
G'=(\{\frac{1}{\lambda} b_r^i, s_i\}),
\]
where the $s_i$'s are defined such that that $s_{i+1} \neq s_i$, for all $i$, and $s_{j_t}=0$.

In words, $G'$ is obtained by ``shrinking'' the search lengths of $G_{b_r}$ by a factor equal to $\lambda$, and by choosing the right parity of branch for starting the search, in a way that, if the hint is trusted, then in $G'$ the searcher will locate the target right as it is about to turn back to $O$ at the end of the $j_t$-th iteration.

Since $G_{b_r}$ is $r$-robust, so is the scalled-down strategy $G'$. It remains then to bound the
consistency $c_{G'}$ of $G'$. Suppose that the hint is trusted. We have that 
\[
d(G',t)=\frac{1}{\lambda}(2\sum_{i=0}^{j_t-1} b_r^i +b_r^{j_t}),
\]
and since $d(t)=b_r^{j_t}/\lambda$ we can bound $c_{G'}$ from above by
\[
\frac{d(G',t)}{d(t)}= 1+2\frac{b_r^{j_t}-1}{b_r^{j_t}(b_r-1)}\leq 1+\frac{2}{b_r-1}=\frac{b_r+1}{b_r-1}. 
\]
We conclude that $G'$ is $(\frac{b_r+1}{b_r-1}, r)$-competitive. 
\end{proof}

Next, we will show that the strategy of Theorem~\ref{thm:position.upper} is Pareto-optimal. To this end, we will need a technical lemma concerning the segment lengths of any $r$-robust strategy.

\begin{lemma}
For $r$-robust strategy $X=(x_i)$, it holds that
\[
x_i \leq (b_r+\frac{b_r}{i+1})x_{i-1},
\]
for all $i \geq 1$, where $x_{-1}$ is defined to be equal to 1.
\label{lemma:helper.bb}
\end{lemma}
\begin{proof}
The proof is by induction on $i$. We first show the claim for $i=0$. Since $X$ is $r$-robust, for a target placed at distance $1+\epsilon$ from the root, it must be that 
\[
1+2x_0 \leq r = 1+2 \frac{b_r^2}{b_r-1} \Rightarrow x_0 \leq \frac{b_r^2}{b_r-1} \leq 2b_r,
\]
where the last inequality follows from the fact that $b_r \geq 2$. Thus, the base case holds.
 
For the induction hypothesis, suppose that the claim holds for all $j \leq i$, that is
$x_{j} \leq (b_r+\frac{b_r}{j+1})x_{j-1}$, for all $j \leq i$. This implies
that
\begin{equation}
x_{i-j} \geq \frac{1}{\prod_{k=0}^{j-1} (b_r+\frac{b_r}{i+1-k})}x_i.
\label{eq:bb.1}
\end{equation}
We will show that the claim holds for $i+1$. From the $r$-robustness of strategy $X$ we have that
\[
\frac{\sum_{j=0}^{i+1} x_j}{x_i} \leq \rho_r \Rightarrow x_{i+1} +\sum_{j=0}^{i-1} x_j \leq (\rho_r-1)x_i,
\]
and substituting $x_0, \ldots ,x_{i-1}$ using~\eqref{eq:bb.1}, we obtain that
\[
x_{i+1} \leq (\rho_r-1-P_i)x_i, \quad \mbox{ where } P_i=\sum_{j=0}^{i-1} \frac{1}{\prod_{k=0}^{j-1} (b_r+\frac{b_r}{i+1-k})}.
\]
It then suffices to show that
\begin{equation}
\rho_r-1-P_i \leq b_r+\frac{b_r}{i+2} \quad \mbox{or equivalently } \ P_i \geq \rho_r-1-b_r\frac{i+3}{i+2}.
\label{eq:bb.2}
\end{equation}
We will prove~\eqref{eq:bb.2} by induction on $i$. For $i=-1$,~\eqref{eq:bb.2} is equivalent to 
\[
2b_r\geq \frac{b_r^2}{b_r-1}-1,
\]
which can be readily verified from the fact that $b_r \geq 2$. Assuming then that~\eqref{eq:bb.2} holds
for $i$, we will show that it holds for $i+1$.
We have 
\begin{align*}
P_{i+1} &=\frac{1}{b_r+\frac{b_r}{i+2}} (1+P_i) \tag{From the definition of $P_i$} \\
&\geq \frac{1}{b_r+\frac{b_r}{i+2}} (1+\rho_r-1-b_r \frac{i+3}{i+2}) \tag{From induction hypothesis}\\
&= \frac{i+2}{b_r(i+3)} (\frac{b_r^2}{b_r-1}-b_r\frac{i+3}{i+2}) \tag{Since $\frac{b_r^2}{b_r-1}=\rho_r$}\\
&> \frac{i+2}{i+3} \frac{b_r}{b_r-1} -1.
\end{align*}
To complete the proof of this lemma, it remains to show that
\[
\frac{i+2}{i+3} \frac{b_r}{b_r-1} -1 \geq \rho_r-1-b_r\frac{i+3}{i+2},
\]
or equivalently, by substituting $\rho_r$ with
the expression $\frac{b_r^2}{b_r-1}$, that
\[
\frac{i+2}{i+3} \frac{1}{b_r-1} +\frac{i+3}{i+2} \geq \frac{b_r}{b_r-1}.
\]
The lhs of the above expression is decreasing in $i$, for every $b_r \geq 2$, thus the lhs is at least
\[
\lim_{i\to \infty} (\frac{i+2}{i+3} \frac{1}{b_r-1} +\frac{i+3}{i+2})=\frac{b_r}{b_r-1},
\]
which concludes the proof.
\end{proof}

We obtain a useful corollary concerning the sum of the first $i-1$ search lengths of an $r$-robust strategy.
\begin{corollary}
For any $r$-robust strategy $X=(x_i)$, it holds that
\[
\sum_{j=0}^{i-1} x_j \geq \frac{x_i}{1+\frac{1}{i+1}}(\frac{b_r}{b_r-1}-\frac{i+2}{i+1}),
\]
and for every $\epsilon \in (0,1]$, there exists $i_0$ such that for all $i>i_0$, $\sum_{j=0}^{i-1} x_j 
\geq (\frac{1}{b_r-1}-\epsilon) x_i$.
\label{cor:sum}
\end{corollary}

\begin{proof}
We have
\begin{eqnarray*}
\sum_{j=0}^{i-1} x_j &=& x_{i-1} +\sum_{j=0}^{i-2} x_j 
\geq x_{i-1}
(1+\sum_{j=0}^{i-2} \frac{1}{\prod_{k=0}^{j-1} 
(b_r+\frac{b_r}{i+1-k})}) \\
&=& x_{i-1}(1+P_{i-1}) \geq x_{i-1}(\rho_r-b_r\frac{i+2}{i+1})  \\
&\geq& \frac{x_i}{b_r+\frac{b_r}{i+1}}(\frac{b_r^2}{b_r-1}-b_r\frac{i+2}{i+1}),
\end{eqnarray*}
where the first inequality follows from Lemma~\ref{lemma:helper.bb}, the second
inequality holds from the property on $P_i$ that was
shown in the proof of Lemma~\ref{lemma:helper.bb}, and the last inequality follows again from
Lemma~\ref{lemma:helper.bb}. 

We now observe that for sufficiently large $i$, the rhs of the inequality is arbitrarily close to $\frac{1}{b_r-1}x_i$, which concludes the proof.
\end{proof}

We can now show a lower bound on the competitiveness of every strategy that matches the upper bound of
Theorem~\ref{thm:position.upper}.

\begin{theorem}
For every $(c,r)$-competitive strategy for linear search in which the hint is the position of the target, it holds that 
$c \geq \frac{b_r+1}{b_r-1} -\epsilon$, for any $\epsilon>0$.
\end{theorem}
\begin{proof}
Let $X=(x_i)$ denote an $r$-robust strategy, with a hint that specifies the position of a target $t$.
Suppose that $X$ locates the target at the $j_t$-th iteration. We have that 
\begin{align*}
c= \frac{d(X,t)}{d(t)}=\frac{2\sum_{i=0}^{j_t-1} x_i+d(t)}{d(t)} 
\geq \frac{2\sum_{i=0}^{j_t-1} x_i+x_{j_t}}{x_{j_t}} 
=1+2\frac{\sum_{i=0}^{j_t-1} x_i}{x_{j_t}}.
\end{align*}

Note that the target $t$ can be chosen to be arbitrarily far from $O$, which means that $j_t$ can be unbounded
(otherwise the strategy would not have bounded robustness). From Corollary~\ref{cor:sum} this implies that 
$\sum_{i=0}^{j_t-1} x_i$ can be arbitrarily close to $x_{j_t}\frac{1}{b_r-1}$, and therefore $c$ is
arbitrarily close to $1+2\frac{1}{b_r-1}=\frac{b_r+1}{b_r-1}$,
which concludes the proof.
\end{proof}

\section{Hint is the direction of search}
\label{sec:direction}

In this section we study the setting in which the hint is related to the {\em direction} of the search. More precisely, the hint is a single bit that dictates whether the target is to the left or to the right of the root $O$. Again, we are interested in Pareto-optimal strategies with respect to competitiveness: namely, for any fixed $r\geq 9$, what is the smallest $c$ such that there exist $(c,r)$-competitive strategies?

A related problem was studied by Schuierer~\cite{schuierer:lower}, which is called {\em biased search}. One defines the {\em left} and {\em right} competitive ratios, as the competitive ratio of a search, assuming that the target hides to the left of the root, or to the right of the root, respectively. However, the searcher does not know the target's branch. Of course we know that the maximum of the left and the right competitive ratios is at least 9 (and for the doubling strategy, this is tight). \cite{schuierer:lower} shows that for any search strategy on the line (not necessarily 9-robust), the {\em average} of the left and the right competitive ratios is at least 9.
At first glance, one may think that this could be an unsurprising, and perhaps even trivial result; 
however this is not the case. The proof in~\cite{schuierer:lower} is not straightforward, and relies in a generalization of a theorem of~\cite{gal:general} which lower bounds the supremum of a sequence of functionals by the supremum of much simpler, geometric functionals. We will discuss this theorem shortly.

The problem studied in~\cite{schuierer:lower} is related to our setting: the left and right competitive ratios correspond to the consistency $c$ and the robustness $r$ of the strategy. Hence from~\cite{schuierer:lower} we know that $c+r \geq 18$. However, there is a lot of room for improvement. In this section we will show a much stronger tradeoff between $c$ and $r$, and we will further prove that it is tight. For example, we will show that 
for any $(c,r)$-competitive strategy, if $c$ approaches 5 from above, then $r$ approaches infinity (in contrast, in this case, the lower bound of~\cite{schuierer:lower} yields $r\geq 13$). In fact, we will show that $c+r$ is minimized when $c=r=9$. To this end, we will apply a parameterized analysis based on Schuierer's approach. We begin with the upper bound, by analyzing a specific strategy. 

\begin{theorem}
For every $b\geq 1$, and $\delta \in (0,1]$, there is a $(c,r)$-competitive strategy for linear search with the hint being the direction of search, in which
\[
c=1+2\cdot (\frac{b^2}{b^2-1}+ \delta \frac{b^3}{b^2-1}) \ \textrm{and} \ 
r=1+2\cdot ( \frac{b^2}{b^2-1}+\frac{1}{\delta} \frac{b^3}{b^2-1}).
\]
\label{thm:direction.upper}
\end{theorem}

\begin{proof}
Suppose, without loss of generality, that the hint points to branch 0. Consider a strategy 
$X=(\{x_i, i \bmod 2\})$,  which starts with branch 0, and alternates between the two branches. 
This strategy has consistency and robustness given by the following expressions, as a consequence 
of~\eqref{eq:comp.X}:
\begin{equation}
c=
 1+2 \cdot \sup_{k \geq 0} \{ \frac{\sum_{i=0}^{2k+1}x_i}{x_{2k}}\} \quad \textrm{and} \quad
r= 1+2\cdot \sup_{k\geq 0} \{ \frac{\sum_{i=0}^{2k}x_i}{x_{2k-1}} \},
\label{eq:ratios.direction.upper}
\end{equation}
where $x_{-1}$ is defined to be equal to 1. 

In addition, the search lengths of $X$ are defined by
\[
x_i= b^i, \textrm{if $i$ even and } \ x_i=\delta b^i, \textrm{if $i$ is odd},  
\]
where we require that $b>1$, and $\delta \in (0,1]$. Note that $X$ is ``biased'' with respect to branch 0, which makes sense since the hint points to that branch.

Substituting these values into~\eqref{eq:ratios.direction.upper}, we obtain that 
\begin{align*}
c&=1+2\cdot \sup_{k \geq 0} \{\frac{\sum_{i=0}^k b^{2i}}{b^{2k}}+ \delta
\frac{\sum_{i=0}^k b^{2i+1}}{b^{2k}}  \} 
=1+2 \cdot \sup_{k \geq 0} \{ \frac{b^{2(k+1)}-1}{(b^2-1) b^{2k}}+ 
\delta  \frac{b^{2k+3}-1}{(b^2-1)b^{2k}}
\} 
\\
& \leq 1+2\cdot (\frac{b^2}{b^2-1}+ \delta \frac{b^3}{b^2-1}).
\end{align*}
Similarly, we have that 
\begin{align*}
r&=1+2\cdot \sup_{k \geq 0} \{\frac{1}{\delta}\frac{\sum_{i=0}^k b^{2i}}{b^{2k-1}}+
\frac{\sum_{i=0}^{k-1} b^{2i+1}}{b^{2k-1}}  \} 
=1+2 \cdot \sup_{k \geq 0} \{ 
\frac{1}{\delta}
\frac{b^{2(k+1)}-1}{(b^2-1) b^{2k-1}}+ 
\frac{b^{2k+1}-1}{(b^2-1)b^{2k-1}}
\} 
\\
& \leq 1+2\cdot ( \frac{1}{\delta} \frac{b^3}{b^2-1}+ \frac{b^2}{b^2-1}).
\end{align*}
\end{proof}

For example, if $\delta=1$, and $b=2$, then Theorem~\ref{thm:direction.upper} shows that there exists a $(9,9)$-competitive strategy. Interestingly, the theorem shows that as the consistency $c$ approaches 5 from above, the robustness $r$ of the strategy must approach infinity. This is because 
the function $\frac{b^2}{b-1}$ is minimized for $b=2$, and hence for $c$ to approach 5 from above, it must be that $b$ approaches 2, and $\delta$ approaches 0. But then $\frac{1}{\delta}$ must approach infinity, and so must $r$.

We will show that the strategy of Theorem~\ref{thm:direction.upper} is Pareto-optimal. To this end, we will use the following theorem of~\cite{schuierer:lower}. Recall the definitions of $\alpha_X$, $X^{+i}$
and $G_a(\gamma_0, \ldots \gamma_{n-1})$ given in Section~\ref{sec:prelim}.

\begin{theorem}[Theorem 1 in~\cite{schuierer:lower}]
Let $p,q$ be two positive integers, and $X=(x_0,x_1,\ldots)$ a sequence of positive numbers with 
$\sup_{n \geq 0} x_{n+1}/x_n <\infty$ and $\alpha_X>0$. Suppose that $F_k$ is a sequence of functionals that 
satisfy the following properties:
\begin{itemize}
\item[(1)]$ F_k(X)$ depends only on $x_0,x_1, \ldots x_{pk+q}$,
\item[(2)] $F_k(X)$ is continuous in every variable, for all positive sequences $X$,
\item[(3)] $F_k(a X)=F_k(X)$, for all $a>0$,
\item[(4)] $F_k(X+Y) \leq \max(F_k(X), F_k(Y))$, for all positive sequences $X,Y$, and
\item[(5)] $F_{k+i}(X) \geq F_k(X^{+ip})$, for all $i \geq 1$.  
\end{itemize}
Then there exist $p$ positive numbers $\gamma_0, \gamma_1, \gamma_{p-1}$ such that 
\[
\sup_{0 \leq k <\infty} F_k(X) \geq \sup_{0 \leq k<\infty} F_k (G_{\alpha_X}(\gamma_0, \ldots ,\gamma_{p-1})). 
\]
\label{thm:sch}
\end{theorem}

We will use Theorem~\ref{thm:sch} to prove a tight lower bound on the competitiveness of any strategy $X$.

\begin{theorem}
For every $(c,r)$-competitive strategy, there exists $\alpha>1$, and $\delta \in (0,1]$ such that 
$c=1+2\cdot (\frac{\alpha^2}{\alpha^2-1}+ \delta \frac{\alpha^3}{\alpha^2-1})$, and 
$r=1+2\cdot ( \frac{\alpha^2}{\alpha^2-1}+\frac{1}{\delta} \frac{\alpha^3}{\alpha^2-1})$.
\label{thm:direction.lower}
\end{theorem}
\begin{proof}
Let $X=(x_0,x_1,\ldots)$ denote a $(c,r)$-competitive strategy, and suppose, without loss of generality, that the hint specifies that the target is in the branch labeled 0. There are two cases concerning $X$: either the first exploration is on the branch labeled 0, or on the branch labeled 1. Let us assume the first case; at the end, we will argue that the second case follows from a symmetrical argument. As we argued in the proof of Theorem~\ref{thm:direction.upper}, in this case the competitiveness of $X$ is described by~\eqref{eq:ratios.direction.upper}. Let us define the functionals
\[
C_k =\frac{\sum_{i=0}^{2k+1}x_i}{x_{2k}} \quad \textrm{and} \ R_k=\frac{\sum_{i=0}^{2k}x_i}{x_{2k-1}}.
\]
Then we have that
\begin{equation}
c=1+2 \cdot \sup_{k \geq 0} C_k \quad \textrm{and} \ r=1+2\cdot \sup_{k\geq 0} R_k.
\label{eq:comp.direction.lower}
\end{equation}
The functional $C_k$ satisfies the conditions of Theorem~\ref{thm:sch} with $p=2$, 
as shown in~\cite{schuierer:lower} therefore there exist
$\gamma_0,\gamma_1 >0$ such that 
\begin{eqnarray*}
\sup_{k \geq 0} C_k &\geq& \sup_{k \geq 0} C_k(G_{\alpha_X}(\gamma_0,\gamma_1))
=\sup_{k\geq 0} \frac{\gamma_0+\gamma_1 \alpha_X+\gamma_0 \alpha_X^2+\ldots +\gamma_0 \alpha_X^2k 
+\gamma_1 \alpha_X^{2k+1}}{\gamma_0 \alpha_X^{2k}} \\
&=& \sup_{k \geq 0} \{ \frac{\sum_{i=0}^k \alpha_X^{2i}}{\alpha_X^{2k}}+
\frac{\gamma_1}{\gamma_0} \frac{\sum_{i=0}^k \alpha_X^{2i+1}} {\alpha_X^{2k}} \}.
\end{eqnarray*}
If $\alpha_X \leq 1$, then the above implies that $\sup_{k\geq 0} C_k=\infty$ (another way of dismissing this case is that if $\alpha_X \leq 1$, then $X$ is bounded and the two branches are not explored to infinity, as required by any strategy of bounded consistency). We can thus assume that $\alpha_X>1$, and we obtain that
\begin{equation}
\sup_{k \geq 0} C_k \geq \sup_{k \geq 0} \{ \frac{\alpha_X^{2k+2}-1}{(\alpha_X^2-1)\alpha_X^{2k}}+
\frac{\gamma_1}{\gamma_0} \frac{\alpha_x^{2k+3}-1}{(\alpha_X^2-1)\alpha_X^{2k}} \}=
\frac{\alpha_X^2}{\alpha_X^2-1} + \frac{\gamma_1}{\gamma_0} \frac{\alpha_X^3}{\alpha_X^2-1}.
\label{eq:R_k}
\end{equation}
We can lower-bound $r$ using a similar argument. The functional $R_k$ satisfies 
the conditions of Theorem~\ref{thm:sch}, again as shown in~\cite{schuierer:lower} therefore  
\begin{eqnarray*}
\sup_{k \geq 0} R_k &\geq& \sup_{k \geq 0} R_k(G_{\alpha_X}(\gamma_0,\gamma_1))
=\sup_{k\geq 0} \frac{\gamma_0+\gamma_1 \alpha_X+\gamma_0 \alpha_X^2+\ldots +\gamma_0 \alpha_X^{2k}}
{\gamma_1 \alpha_X^{2k-1}} \\
&=& \sup_{k \geq 0} \{ \frac{\gamma_0}{\gamma_1} \frac{\sum_{i=0}^k \alpha_X^{2i}}{\alpha_X^{2k-1}}+
 \frac{\sum_{i=0}^k \alpha_X^{2i-1}} {\alpha_X^{2k-1}}\}.
\end{eqnarray*}
Using the same argument as earlier, it suffices to consider only the case $\alpha_X>1$, in which case
we further obtain that
\begin{equation}
\sup_{k \geq 0} R_k \geq \sup_{k \geq 0} \{ \frac{\gamma_0}{\gamma_1} \frac{\alpha_X^{2k+2}-1}{(\alpha_X^2-1)\alpha_X^{2k-1}}+
\frac{\alpha_x^{2k+1}-1}{(\alpha_X^2-1)\alpha_X^{2k-1}} \}=
\frac{\gamma_0}{\gamma_1} \frac{\alpha_X^3}{\alpha_X^2-1} +  \frac{\alpha_X^2}{\alpha_X^2-1}.
\label{eq:W_k}
\end{equation}
Let us define $\delta=\frac{\gamma_1}{\gamma_0}>0$. The result follows then by combining~\eqref{eq:comp.direction.lower},~\eqref{eq:R_k} and~\eqref{eq:W_k}. Note that if we require that $c\leq r$, it must be that $\delta\leq 1$, 
since $\alpha_X>1$.

It remains to consider the symmetric case, in which in $X$, the first explored branch is branch 1. 
In this case the analysis is essentially identical: in~\eqref{eq:comp.direction.lower} we substitute $C_k$ with
$R_k$ and vice versa, in the expressions of $c$ and $r$, and in the resulting lower bounds we require that 
$\delta>1$.
\end{proof}

It is important to note that in the proof of Theorem~\ref{thm:direction.lower} we used the fact that the values
$\gamma_0$ and $\gamma_1$ depend only on $X$ and not on any functionals defined over $X$, as follows from the proof of Theorem~\ref{thm:sch} in~\cite{schuierer:lower}.

Theorem~\ref{thm:direction.lower} implies that any $(c,r)$-competitive strategy $X$ is such that 
\[
c+r \geq 2+4\frac{\alpha_X^2}{\alpha_X^2-1} +2(\delta+\frac{1}{\delta}) \frac{\alpha_X^3}{\alpha_X^2-1},
\]
which is minimized at $\delta=1$, for which we obtain that $c+r$ is minimized only if
$c=r= 1+2\frac{\alpha_X^2}{\alpha_X-1} \geq 9$. We conclude that the average of a strategy's consistency and robustness (or the average of the left and right competitive ratio, in the terminology of~\cite{schuierer:lower}) is minimized only by strategies that are 9-robust.

\section{Hint is a $k$-bit string}
\label{sec:string}

In this section we study the setting in which the searcher has access to a {\em hint string} of $k$ bits. We first consider the case $k=1$. In Section~\ref{subsec:k} we will study the more general case.

It should be clear that even a single-bit hint is quite powerful, and that the setting is non-trivial. For example, the bit can indicate the right direction for search, as discussed in Section~\ref{sec:direction}, but it allows for other possibilities, such as whether the target is at distance at most $D$ from the root, for some chosen $D$. The latter was studied in~\cite{jaillet:online}, assuming that the hint is correct. More generally, the hint can induce a partition of the infinite line into two subsets $L_1$ and $L_2$, such that the hint dictates whether the target is hiding on $L_1$ or $L_2$. 

We begin with the upper bound, namely we describe a specific search strategy, and the corresponding hint bit (as well as the query which it responds).  
Consider two strategies of the form
\[
X_1=(b_r^i) \ \textrm {and} \ X_2=(b_r^{i+\frac{1}{2}}).
\]
Note that both $X_1$ and $X_2$ are $r$-robust: $X_1$ is geometric with base $b_r$, whereas $X_2$ is obtained from $X_1$ by scaling the search lengths by a factor equal to $b_r^{1/2}$.
We also require that the two strategies start by searching the same branch, hence in every iteration, they likewise search the same branch. 

We can now define a strategy $Z$ with a single bit hint, which indicates whether the searcher should choose strategy $X_1$ or strategy $X_2$. For any given target, one of the two strategies will outperform the other, assuming the hint is trusted. Thus, an equivalent interpretation of the hint is in the form of a partition of the infinite line into two sets $L_1$ and $L_2$, such that if the target is in $L_i$, then $X_i$ is the preferred strategy, with $i \in [1,2]$. See Figure~\ref{fig:1bit} for an illustration.

\begin{figure}
   \centerline{\includegraphics[height=5cm, width=15cm]{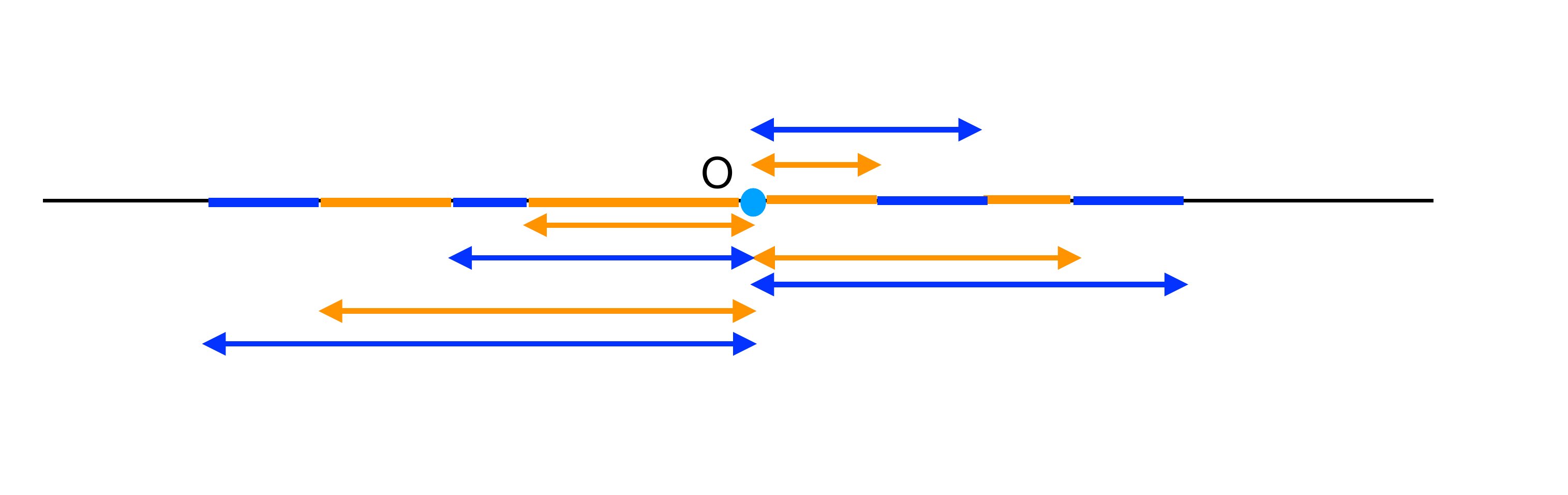}}
  \caption{Illustration of strategy $Z$, using the first four segments of strategies $X_1$ and $X_2$. Blue (dark)and orange (faded) segments correspond to the search segments of strategies $X_1$ and $X_2$, respectively. The parts of the line in blue (resp. orange) indicate the hiding intervals for the target such that $X_1$ (resp. $X_2$) is preferred, and thus chosen by the hint.}
  \label{fig:1bit}
\end{figure}

The following result bounds the performance of this strategy, and its proof will follow as a corollary of a more general theorem concerning $k$-bit strings that we show in Section~\ref{subsec:k} (Theorem~\ref{prop:k.upper}).
\begin{proposition}
For given $r \geq 9$, the above-defined strategy $Z$ is $r$-robust and 
has consistency at most $1+2\frac{a^{3/2}}{a-1}$, where
\[
      a = \left\{
      \begin{array}{ll}
        b_r
        & \mbox{if } r \leq 10
        \\[1em]
        3
        &
        \mbox{if }  r \geq 10.
      \end{array}
      \right.
  \]
\label{prop.single.upper}  
\end{proposition}

Note that if $r=9$, then $Z$ has consistency $1+4\sqrt{2} \approx 6.657$. For $r \in [9,10]$, the consistency of $Z$ is decreasing in $r$, as one expects. For $r \geq 10$, the consistency is $1+3\sqrt{3} \approx 6.196$.

We now turn our attention to lower bounds. To this end, we observe that a single-bit hint $h$ has only the power to differentiate between two {\em fixed} strategies, say  $X=(x_i)$, and $Y=(y_i)$, i.e., two strategies that are not defined as functions of $h$. We say that $Z$ is {\em determined} by strategies $X$ and $Y$, and the bit $h$. 

\paragraph{Setting up the lower-bound proofs} 
We give some definitions and notation that will be used in the proofs of Theorems~\ref{thm:bit.9.upper} and~\ref{thm:asymptotic}. Let $Z$ be determined by strategies $X$ and $Y$, and a single-bit hint $h$. Let $C$ denote the lower bound on the consistency of $Z$ that we wish to show. 
For given $i$, define $T_X^i=2\sum_{j=0}^{i}x_j+x_{i-1}$, and similarly for $T_Y^i$. Define also $q=r/C$.

Note that a searcher that follows strategy $X$ will turn towards the root at iteration $i-1$, after having explored some branch $\beta_i \in \{0,1\}$ up to distance $x_{i-1}$. Thus, $X$ barely misses a target that may be hiding at branch $\beta_i$, and at distance $x_{i-1}+\epsilon$ from $O$, with $\epsilon>0$ infinitesimally small, and thus requires time $T_X^i$ to discover it. We will denote this hiding position of a potential target by $P_i$. If, on the other hand, the searcher 
follows $Y$, then it can locate a target at position $P_i$ at a time that may be smaller than $T_X^i$; 
let $\tau_i$ denote this time. When strategy $Y$ locates a target hiding at $P_i$, it does so by exploring branch $\beta_i$ to a length greater than $x_{i-1}$. Let $j_i$ be the iteration  at which $Y$ locates $P_i$, thus $y_{j_i} \geq x_{i-1}$. Last, let $Q_i$ denote the position in branch $\beta_i$ and at distance $y_{j_i}+\epsilon$ from $O$. In words, if a target hides at $Q_i$, then strategy $Y$ barely misses it when executing the search segment $y_{j_i}$.

We first show a lower bound on the consistency of 9-robust strategies. In the proof we will not replace all parameters with the corresponding values (e.g., we will sometimes use $r$ to refer to robustness, instead of the value 9). We do so because the arguments in the proof can be applied to other settings, as will become clear in the proof of Theorem~\ref{thm:asymptotic}.

\begin{theorem}
For any $(c,9)$-competitive strategy with single-bit hint, it holds that $c\geq 5$.
\label{thm:bit.9.upper}
\end{theorem}

\begin{proof}
We will prove the result by way of contradiction. Let $C=5$, and suppose that there is a strategy $Z$ of consistency strictly less than $C$. Let $Z$ be determined by two fixed strategies $X$ and $Y$. Both $X$ and
$Y$ must be $r$-robust (i.e., 9-robust), otherwise $Z$ cannot be $r$-robust.

Fix $i_0 \in \mathbb{N}$. Suppose first that $i_0$ is such that for all $i \geq i_0$, we have  $\tau_i \geq \frac{1}{q} T_X^i$. In this case, for a target at position $P_i$, defined earlier, $X$ requires time $T_X^i$ to locate it, whereas $Y$ requires time at least $\tau_i \geq \frac{1}{q} T_X^i$ to locate it, thus the minimum time $X$ or $Y$ can locate this target is $\frac{1}{q} T_X^i$. Therefore, the consistency of $Z$ is at least
\begin{equation}
c \geq \sup_{i \geq i_0} \frac{T_X^i}{q \cdot x_{i-1}}= \frac{1}{q} \sup_{i \geq i_0} \frac{T_X^i}{x_{i-1}} \geq \frac{C}{r} \cdot  r =C,
\label{eq:supremum.9} 
\end{equation}
which is a contradiction. Here, we used crucially the fact that $\sup_{i \geq i_0} \frac{T_X^i}{x_{i-1}} \geq 9$, for any 9-robust strategy $X$ and any $i_0$\footnote{In general, this statement is not immediately true for arbitrary $r>9$.}. Specifically, there exists sufficiently large $i$ such that  $T_X^i$ is arbitrarily close to $9x_{i-1}$.

It must then be that $i_0$ does not obey the property described above, namely for some $i\geq i_0$ we have that
$\tau_i \leq \frac{1}{q} T_X^i$. Since $X$ is $r=9$-robust, it must also be that $\frac{T_X^{i}}{x_{i-1}}\leq r$, as can be seen if a target hides at $P_i$. We therefore obtain that
\begin{equation}
\tau_i \leq \frac{1}{q} T_X^i \leq \frac{1}{q} \cdot r \cdot x_{i-1}= C x_{i-1}.
\label{eq:tau.9.upper}
\end{equation}
We can also give a lower bound on $\tau_i$, as follows. Recall that we denote by $y_{j_i}$ the segment at which strategy $Y$ locates a target at position $P_i$. For arbitrarily small $\epsilon>0$, we can choose $i_0$ sufficiently large, which also implies that $j_i$ can also be sufficiently large (since otherwise $Y$ would not have finite robustness), so that Corollary~\ref{cor:sum} applies. To simplify the arguments, in the remainder of the proof we will assume that the corollary applies with $\epsilon=0$; this has no effect on correctness, since we want to show a lower bound of the form $C-\delta$ on the consistency, and $\epsilon$ can be made as small as we want in comparison to $\delta$.  More precisely, we obtain that
\begin{equation}
\tau_i =2 \sum_{l=0}^{j_i} y_l +x_{i-1} \geq \frac{2}{b_r-1}y_{j_i} +x_{i-1}.
\label{eq:tau.9.lower}
\end{equation}
Combining~\eqref{eq:tau.9.upper} and~\eqref{eq:tau.9.lower} we have
\begin{equation}
y_{j_i} \leq \frac{C-1}{2}(b_r-1) x_{i-1}.
\label{eq:l_i.9}
\end{equation}
In particular, since $r=9$ and $C=5$, we have that $y_{j_i} \leq 2 x_{i-1}$.

Consider now a target at position $Q_i$, and recall that this position is at distance infinitesimally larger 
than $y_{j_i}$. We will show that in both $X$ and $Y$, there exists an $i\geq i_0$ such that the searcher walks distance at least $C \cdot y_{j_i}$ before reaching this position, which implies that $Z$ has consistency at least $C$, and which yields the contradiction.

Consider first strategy $Y$. In this case, the searcher walks distance at least $T_Y^{j_i}$, to reach $Q_i$, from the definition of $T_Y$. Since $r=9$, we know that $\sup_{i\geq i_0} \frac{T_Y^{j_i}}{y_{j_i}}\geq 9$, for any $i_0$, hence there exists an $i \geq i_0$ such that the distance walked by the searcher is at least $r \cdot y_{j_i}$, and hence at least $C \cdot y_{j_i}$.

Consider now strategy $X$. In this case, in order to arrive at position $Q_i$, the searcher needs to walk distance 
$T_X^i$, then at least an additional distance $y_{j_i}-x_{i-1}$ to reach $Q_i$.
Let us denote by $D_X^i$ this distance. We have
\begin{align}
D_X^i &= T_X^i +y_{j_i}-x_{i-1} \\
&\geq 9 x_{i-1}+y_{j_i}-x_{i-1} \tag{From Corollary~\ref{cor:sum} and since $T_X^i$ is arbitrarily close to
$9x_{i-1}$} \\
&= 8x_{i-1} +y_{j_i}.
\label{eq:D_X.9}
\end{align}

We then bound the ratio $D_X^i /y_{j_i}$ from below as follows.
\begin{align}
\frac{D_X^i}{y_{j_i}} &\geq \frac{8x_{i-1} +y_{j_i}}{y_{j_i}}  \label{eq:final.9}
\\
&= 1+ 8 \frac{x_{i-1}}{y_{j_i}}  \nonumber  \\
&\geq 1+ 8 \frac{x_{i-1}}{2 \cdot x_{i-1}} 
\tag{From the fact that $y_{j_i} \leq 2 \cdot x_{i-1}$} \nonumber \\
&=5.
\end{align}
We thus conclude that $C\geq 5$, which yields the contradiction, and completes the proof.
\end{proof}

Showing a lower bound on the consistency, as a function of general $r>9$ is quite hard, even for the case of 
a single-bit hint. The reason is that as $r$ increases, so does the space of $r$-robust strategies. For example,
any geometric strategy $G_b$ has robustness $r$, as long as $b\in [\frac{\rho_r-\sqrt{\rho_r^2-4\rho_r}}{2}, \frac{\rho_r+\sqrt{\rho_r^2-4\rho_r}}{2}]$. In what follows we will show a lower
bound for a class of strategies which we call {\em asymptotic}. More precisely, recall the definition of $T_X^i$. We call an $r$-robust strategy $S$
{\em asymptotic} if $\sup_{i \geq i_0} \frac{T_X^i}{x_{i-1}}=r$, for all fixed $i_0$. In words, in an asymptotic strategy, the worst-case robustness (i.e., the worst case competitive ratio without any hint) can always be attained by targets placements sufficiently far from the root. All geometric strategies, including the doubling strategy, have this property, and this holds for many strategies that solve search problems on the line and the star, such as the ones described in the introduction. 
Note also that the strategies $X_1$ and $X_2$ that determine the strategy $Z$ in the statement of Proposition~\ref{prop.single.upper} are asymptotic, since they are near-geometric. Thus, the lower bound we show in the next theorem implies that in order to substantially improve consistency, one may have to resort to much more complex, and most likely irregular strategies.

\begin{theorem}
Let $Z$ denote a strategy with 1-bit hint which is determined by two $r$-robust, asymptotic strategies $X$ and
$Y$. Then $Z$ is $(c,r)$-competitive, with 
$c \geq 1+\frac{2b_r}{b_r-1}$. 
\label{thm:asymptotic}
\end{theorem}

\begin{proof}
We show how to modify the proof of Theorem~\ref{thm:bit.9.upper}. 
Let $C$  be equal to $1+\frac{2b_r}{b_r-1}$, and suppose, by way of contradiction, that the robustness of $Z$ is strictly less than $C$. 
First, we note that~\eqref{eq:supremum.9} applies since $X$ is asymptotic, and so do 
equations~\eqref{eq:tau.9.upper},~\ref{eq:tau.9.lower} and~\eqref{eq:l_i.9}. 

As in the proof of Theorem~\ref{thm:bit.9.upper}, we next consider a target hiding at position $Q_i$. We then can argue that there exists $i \geq i_0$ such that, for this hiding position, the total distance walked by $Y$ is at least $r \cdot y_{j_i} \geq C y_{j_i}$. Here we use the fact that $Y$ is asymptotic.

Next we consider strategy $X$, and we bound its cost for locating the hiding position $Q_i$. We have, similarly to the proof of 
Theorem~\ref{thm:bit.9.upper}, that

\begin{align}
D_X^i = T_X^i +y_{j_i}-x_{i-1} 
\geq (r-1)x_{i-1}+ y_{j_i}.
\tag{Since $X$ is asymptotic} \nonumber
\label{eq:D_X.9.asymptotic}
\end{align}

Therefore
\begin{align*}
R \geq \frac{D_X^i}{y_{j_i}}  \geq 1+\frac{(r-1)x_{i-1}}{y_{j-1}} \geq 1+\frac{2(r-1)}{(C-1)(b_r-1)},
\end{align*}
where the last inequality follows from~\eqref{eq:l_i.9}. Solving this inequality for $C$ we obtain
\[
C \geq 1+\sqrt{\frac{2(r-1)}{b_r-1}},
\]
and substituting with $r=1+2\frac{b_r^2}{b_r-1}$ we obtain that $C \geq 1+\frac{2b_r}{b_r-1}$, a contradiction, and the proof is complete.
\end{proof}

\subsection{$k$-bit hints}
\label{subsec:k}

Here we consider the general setting in which the hint is a $k$-bit string, for some fixed $k$. First, we give an upper bound that generalizes Proposition~\ref{prop.single.upper}. We will adapt an algorithm proposed in~\cite{DBLP:conf/innovations/0001DJKR20} for the {\em online bidding} problem with untrusted advice. 
Consider $2^k$ strategies $X_0, \ldots ,X_{2^k-1}$, where 
\[
X_j = (a^{i+\frac{j}{2^k}})_{i\geq 0},
\]
for some $a$ to be determined later, and where all the $X_j$ have the same parity: they all search the same branch in their first iteration and, therefore in every iteration as well. Define a strategy $Z$, which is determined by
$X_0, \ldots ,X_{2^k-1}$, and in which the $k$-bit hint $h$ dictates the index $j$ of the chosen strategy $X_j$. 
In other words, $h$ answers the query $Q_h$=``which strategy among $X_0, \ldots ,X_{2^k-1}$ should the searcher choose?''. An equivalent interpretation is that the statements of the $X_j$'s induce a partition of the line, such that for every given target position, one of the $X_j$'s is the preferred strategy. Thus every bit $i$ of the hint can be thought, equivalently, as the answer to a partition query $Q_i$ of the line, i.e., of the form ``does the target belong in a subset $S_i$ of the line or not?''. 

\begin{theorem}
For every $r\geq 9$, the strategy $Z$ defined above is $(c,r)$-competitive
with $c\leq1+2\frac{a^{1+\frac{1}{2^k}}}{a-1}$, where
\[
      a= \left\{
      \begin{array}{ll}
        b_r
        & \mbox{if } \rho_r \leq \frac{(1+2^k)^2}{2^k}
        \\[1em]
        
        1+2^k,
        &\mbox{if }  \rho_r \geq \frac{(1+2^k)^2}{2^k}.
      \end{array}
      \right.
 \]

\label{prop:k.upper}
\end{theorem}

\begin{proof}
For a given target $t$, let $j_t=h$, that is, $j_t$ is the index of the best strategy among the $X_j$'s for located $t$, as dictated by $h$. From the statements of the strategies, this implies that there exists some $i_t$ such that
\[
a^{i_t+\frac{j_t}{2^k}-1-\frac{1}{2^k}} < d(t) \leq a^{i_t+\frac{j_t}{2^k}},
\]
Then 
\[
\frac{d(X_{j_t},t)}{d(t)} = \frac{2\sum_{l=0}^{i_t-1}a^{l+\frac{j_t}{2^k}}+d(t)} {d(t)} \leq
1+2 \frac{\sum_{l=0}^{i_t-1} a^{l+\frac{j_t}{2^k}} }
{a^{i_t+\frac{j_t}{2^k}-1-\frac{1}{2^k}}} \leq 1+2\frac{a^{1+\frac{1}{2^k}}}{a-1}.
\]
It remains to chose the right value for $a$. Recall that the geometric strategy $G_a$ is 
$(1+2\frac{a^2}{a-1})$-robust, thus so is every $X_j$ defined earlier. We then require that
$a$ is such that 
\[
1+2\frac{a^2}{a-1} \leq r=1+2\rho_r \Rightarrow \frac{a^2}{a-1} \leq \rho_r.
\]
We also require that the consistency of $Z$, as bounded earlier, is minimized. Therefore, we seek 
$a>1$ such that $\frac{a^{1+\frac{1}{2^k}}}{a-1}$ is minimized and $\frac{a^2}{a-1} \leq \rho_r$.
Using simple calculus, we obtain that the value of $a$ that satisfies the above constraints is as in the
statement of the theorem. 
\end{proof}

For example, for $r=9$, we obtain a strategy that is $(1+2^{2+\frac{1}{2^k}},9)$-competitive. Thus, the consistency decreases rapidly, as function of $k$, and approaches 5. 

Last, it is easy to see that no 9-robust strategy with hint string of any size can have consistency better than 3. To see this, let $X$ be any 9-robust strategy, and let $i_t$ be the iteration in which it locates a target $t$. Since $t$ can be arbitrarily far from $O$, $i_t$ is unbounded, and thus Corollary~\ref{cor:sum} applies. We thus have that
\[
\frac{d(X,t)}{d(t)} =\frac{\sum_{j=0}^{i_t-1}x_j+d(t)}{d(t)}=1+2\frac{\sum_{j=0}^{i_t-1}x_j}{d(t)}
\geq 1+2\frac{\sum_{j=0}^{i_t-1}x_j}{x_{i_t}} \geq 1+2 \frac{(1-\epsilon)x_{i_t}}{x_{i_t}},
\]
thus the consistency of $X$ cannot be smaller that $3-\epsilon$, for any $\epsilon$. The same holds then for any strategy that is determined by any number of $9$-robust strategies, and thus for any hint.

\bibliographystyle{plain}
\bibliography{targets}

\end{document}